\documentclass[12pt]{article}
\usepackage{amsmath,amssymb,amsthm}
\usepackage[margin=1in]{geometry}
\usepackage{booktabs}
\usepackage{enumitem}
\usepackage{setspace}
\usepackage{booktabs,array}
\newcolumntype{L}[1]{>{\raggedright\arraybackslash}p{#1}}
\usepackage{float}
\usepackage[hidelinks]{hyperref}
\usepackage[title]{appendix}

\newtheorem{theorem}{Theorem}
\newtheorem{lemma}{Lemma}
\newtheorem{proposition}{Proposition}
\newtheorem{corollary}{Corollary}
\newtheorem{definition}{Definition}

\newtheorem{example}{Example}
\newtheorem{fact}{Fact}

\newcommand{\C}{\mathcal C}
\newcommand{\M}{\mathcal M}
\newcommand{\T}{\Theta}
\newcommand{\Lcal}{\mathcal L}
\newcommand{\preceqM}{\preceq_M}

\title{Certifying Priority in Evidence-Based Allocation}
\author{{\large Charles Po-Cheng Huang\thanks{Department of Economics, National Central University, cpchuang@ncu.edu.tw} \:\:\: Gunhaeng Lee\thanks{Institute of Economics, Academia Sinica, glee@econ.sinica.edu.tw}\vspace*{1ex}}}
\date{September 15, 2026}

\begin{document}
\maketitle

\begin{abstract}
Many allocation rules give priority to applicants satisfying favored
conditions, including disability, homelessness, and veteran status. Because
these conditions are not directly observed, a rule ranking applicants by
characteristics must first determine what submitted certificates prove. This
paper studies allocation with hard but partial evidence: a disability document may prove disability without proving the absence of a disqualifying tenancy record. A claim is certifiable exactly when every intended claimant has a feasible certificate excluding all outsiders. Chain rules---scores,
lexicographic orders, and scalar rankings---require certification of every priority cutoff, whereas reserve rules require certification of each eligibility class. Allocation rules thus select the claim family to be certified. When certification fails, canonical repairs demote unsupported chain priorities and delete unsupported reserve labels. Nonempty reserve eligibility classes can be represented as upper contours of a single chain if and only if they are nested. A housing application illustrates the consequences.
\end{abstract}


\newpage

\section{Introduction}

Priority rules for scarce public resources are usually written in terms of applicant characteristics. A housing authority may favor applicants who are disabled, homeless, veterans, have qualifying family status, or have no disqualifying tenancy record. But these characteristics are private. The authority sees only the documents, records, database responses, and other evidence that applicants submit. Hard evidence guarantees the authenticity of submitted content; it does not tell the authority whether the applicant has submitted everything relevant. Even a requirement to ``submit all evidence'' is therefore ineffective when the authority cannot verify that the submission is complete. Before an institution can rank applicants or reserve capacity for them, it must decide what their evidence is sufficient to prove.

This paper asks which priority rules can be supported by partial hard evidence. The answer depends on the architecture of the allocation rule. A \emph{chain priority system}, such as a score, a lexicographic rule, or another scalar ranking, places all applicants on one ordered scale. A chain only needs to verify an applicant's position relative to each cutoff; that is, every applicant belonging at or above a cutoff must hold a certificate that no applicant below it could have produced. The set of applicants at or above a given cutoff is that cutoff's \emph{upper contour}, so the chain's evidentiary burden is to certify these upper contours, one for each cutoff. A \emph{reserve system} instead preserves separate eligibility labels and assigns label-specific capacities. Its burden is to certify each reserve's eligibility class. Because the two architectures ask applicants to prove different things, the same applicant characteristics and evidence technology may support one and not the other. The paper characterizes this compatibility and constructs a canonical conservative evidence-safe revision when the target rule is not compatible with the available evidence. Allocation design is therefore also evidence design: the rule determines not only who should receive priority, but also which claims applicants must be able to establish.

We formalize this problem using a finite hard-evidence technology. An applicant's private type records the categories to which she belongs. Each type can submit certain certificate contents and cannot fabricate any other content, but may withhold feasible evidence. Every content therefore has a \emph{certifier set}: the set of types capable of submitting it. A type class is content-certifiable if the planner can accept a set of contents that admits all and only the types in that class. The basic criterion is simple. A class is certifiable if and only if every intended claimant has some feasible content that no outsider can submit, or equivalently, if the class is a union of certifier sets. Complete type revelation is unnecessary; what matters is whether the evidence separates the claim relevant for allocation. At the same time, several separately informative contents need not certify their conjunction since an applicant can already pass a certificate test by submitting any one accepted content. Joint claims require a jointly readable content. Under a full-report normality condition, this problem disappears, and certifiability reduces to monotonicity with respect to evidence mimicry. 

Our first main result characterizes chain priority systems (Theorem~\ref{thm:chain-burden}). A target chain is safely implementable if and only if every proper upper contour of its ranking is content-certifiable. The canonical implementation uses a skeptical reading: each submitted content is assigned the lowest priority held by any type capable of submitting it. A type attains its target class only when some feasible content excludes every type below that class. The reason upper contours, rather than exact classes, matter is scarcity. At a binding priority boundary, an applicant needs to establish that she belongs on or above the relevant side of the cutoff. Remaining pooled with higher-priority types is harmless; remaining pooled with a lower-priority type can change who receives the marginal unit. The same condition is also necessary and sufficient for dominant-strategy implementation of the full-information chain allocation. Hence, allowing a mechanism to condition on the entire profile of submitted contents does not overcome an uncertifiable priority boundary. 

Our second main result gives the parallel characterization for reserve systems (Theorem~\ref{thm:reserve-burden}). A reserve asks whether an applicant is inside the eligibility class since each label has its own capacity. Under label-specific applications, a target eligibility map is safely implementable if and only if every reserve eligibility class is content-certifiable. The skeptical reserve reading assigns a label to a content only when every type capable of submitting that content is entitled to the label. This eligibility characterization does not depend on how the downstream rule rations an oversubscribed reserve, allocates open capacity, or treats applicants eligible for several reserves. When the condition holds, applying for every certifiable label is weakly dominant and reproduces the full-information allocation under any feasible, label-monotone downstream rule. The application protocol nevertheless matters. If the institution requires one common content to establish an applicant's entire label set, that single content must jointly establish all and only those labels; separate contents for separate labels no longer suffice. 

The third main result addresses targets that fail these characterizations (Theorem~\ref{thm:repairs}). Holding the evidence technology fixed, we derive a canonical conservative repair for each architecture. For a chain, a skeptical closure operator makes the smallest pointwise demotions needed for safe implementation. For reserves, a certifiable-interior operator deletes unsupported labels and retains, label by label, the largest certifiable part of each target eligibility class. Both repairs are computed in one step and are fixed points of their respective operators. The original rule is safely implementable exactly when the repair leaves it unchanged. A repair does not recover an infeasible target allocation; it identifies the pointwise closest conservative revision of the target that the available evidence can support. 

The applications show why the content of evidence, rather than simply its amount, matters. Positive evidence can establish membership in a favored category but cannot establish the absence of an adverse condition. Negative evidence can establish clearance but not favorable membership. A rule that rewards both---for example, disability together with a clean tenancy record---requires a signed or otherwise jointly informative content that certifies presence and absence together. In the housing application, a disability document alone cannot justify clean-history points, and a tenancy clearance alone cannot justify disability points. The framework identifies precisely which priority claims must be weakened when the institution cannot read these facts jointly. 

The analysis also separates chain and reserve systems even when all relevant claims are certifiable. Nonempty reserve eligibility classes can be represented as upper contours of a single chain if and only if they are nested by inclusion. In the housing application, disability and homelessness reserves with the same clearance requirement cross: neither eligibility class contains the other. No score, lexicographic rule, or other scalar ranking can represent both classes through priority cutoffs. For every proposed chain, at least one class generates a single-reserve contest whose allocation differs from the chain. This disagreement is not caused by deficient evidence or poorly chosen score weights; it reflects the restriction to one ordered priority scale.

\paragraph{Related literature.}
The paper relates to work on hard evidence, partial provability, mechanisms
with evidence, certifiable communication, evidence-reading mechanisms, and
falsification-proof non-market allocation.
Green and Laffont (1986), Bull and
Watson (2007), Ben-Porath and Lipman (2012), and Ben-Porath, Dekel, and
Lipman (2019) study mechanism design when agents can prove some facts but not
others. Ben-Porath, Dekel, and Lipman (2014) study optimal allocation when the principal can verify agents' private information at a cost. Hagenbach, Koessler, and Perez-Richet (2014) study certifiable
pre-play communication and full disclosure. Koessler and Perez-Richet (2019)
study reading mechanisms that apply a social choice function to a consistent
interpretation of evidence. Perez-Richet and Skreta (2026) study optimal
non-market allocation when eligibility-relevant scores are manipulable and
falsification is costly. Their problem is to design rules that deter score
manipulation; ours is to characterize the proof burden imposed by a fixed
allocation rule under partial hard evidence. This paper specializes the
reading problem to scarce-unit allocation and shows how allocation rules select
the claim family: upper contours for chains, eligibility labels for reserves,
and joint presence--absence bundles for signed rules.

The paper also relates to the market design literature on allocating scarce
capacity under priorities. Hafalir, Yenmez, and Yildirim (2013) and Echenique and Yenmez (2015) study affirmative action and controlled-choice constraints in school assignment. Kominers and S\"onmez (2016) analyze slot-specific priorities, and Dur, Kominers, Pathak, and S\"onmez (2018) show how the precedence order over reserved seats shapes the resulting allocation. S\"onmez and Yenmez (2022) study vertical, horizontal, and overlapping reservations in Indian affirmative action, and Pathak, S\"onmez, \"Unver, and Yenmez (2024) design reserve systems for pandemic rationing of vaccines, ventilators, and antiviral treatments. These papers take the priority structure or the eligibility categories as given and study the allocation rule built on top of them. 
We ask the prior question of how such a priority can be formed at all. Because applicants are endowed with privately known categorical types, the
planner learns categories only through the certificates applicants choose to
disclose, so evidence disclosure is the first step in forming the priority. The priority structure those papers take as a primitive is, in our model, an object that the evidence technology must be able to support.

\section{Evidence and Certifiable Claims}

\subsection{Evidence Technology}

There is a finite set of categories $\C=\{C_1,\ldots,C_K\}.$ An applicant's type is the set of categories to which she truly belongs. We write $\theta\subseteq\C$ and let $\T=2^\C$ denote the finite type space. Types
are private information. The inclusion order on category sets is the natural partial order on $\T$: $\theta\subseteq\theta'$ means that type $\theta'$ belongs to every category that type $\theta$ belongs to, and possibly more.

A certificate is a verifiable document, message, or claim submitted by an
applicant. An allocation rule uses only the certificate's relevant content.
Two documents may have different labels, issuers, or formats but verify the
same facts. A content-based rule treats such documents alike. If labels or
issuers carry independent information and the planner is allowed to use it,
the relevant content space should include those labels.

Let $\M$ be the finite set of certificate contents, with a typical element
$m$. The evidence technology is a correspondence $M: \T \rightrightarrows \M$, and $M(\theta)$ gives the contents feasible for type $\theta$. 
Evidence is hard: type $\theta$ can submit any $m\in M(\theta)$ and cannot produce any content outside $M(\theta)$. 
Because applicants may withhold evidence, $M(\theta)$ may contain contents that reveal only part of what the type could prove. A content is the
evidentiary object that the planner evaluates jointly. It may be a single
physical document, a database response, or a portfolio of documents. Thus,
the evidence technology determines which individual documents or portfolios
are available as jointly readable contents in $\M$.

Each content $m$ defines its \emph{certifier set}
\[
  M^{-1}(m)=\{\theta\in\T:m\in M(\theta)\},
\]
the set of types that could have submitted it. The certifier set describes
what the content rules out: $m$ separates $\theta$ from $\theta'$ if
$\theta\in M^{-1}(m)$ and $\theta'\notin M^{-1}(m)$. A content is more
informative when its certifier set is smaller.

We use two harmless feasibility conventions. Each type has at least one
feasible content, and every content is feasible for at least one type. Contents
with empty certifier sets can be deleted, and a type with no feasible content
can be accommodated by adding a null content available to every type.

\subsection{Content-Certifiability}

At the certification stage, the planner may need to verify a binary claim
about one applicant; for example, that she is disabled and has no
disqualifying tenancy record. Satisfying both conditions, only one of them, or
neither may place her in a different priority class. We represent each such
claim by a type class $D\subseteq\T$. The class formulation is
intentional. A certificate may establish a
priority-relevant property without identifying the applicant's entire type.
Exact type certification is the special case $D=\{\theta\}$.

\begin{definition}\label{def:certificatetest}
A \emph{certificate test} is a set $A\subseteq\M$ of contents that the planner accepts as establishing a claim. The test is \emph{complete} for
$D\subseteq\T$ if
\[
  M(\theta)\cap A\neq\emptyset
  \qquad\text{for every }\theta\in D,
\]
and \emph{sound} for $D$ if
\[
  M(\theta)\cap A=\emptyset
  \qquad\text{for every }\theta\notin D.
\]
The test certifies $D$ if it is both complete and sound, or equivalently, if
\[
  D=\{\theta\in\T:M(\theta)\cap A\neq\emptyset\}.
\]
We say that $D$ is \emph{content-certifiable} if some certificate test
certifies it.
\end{definition}

In classification terms, completeness rules out false exclusions of intended
claimants, while soundness rules out false inclusions of unintended claimants.

\begin{example}[When separate proofs do not certify a conjunction]
\label{ex:separate-documents-no-conjunction}
Let $\C=\{d,t\}$, where $d$ denotes disability status and $t$ denotes a
disqualifying tenancy record. The claim ``disabled and clear'' corresponds to
the singleton class
\[
  D^{\mathrm{clear}}=\{\{d\}\}.
\]

Consider the content space
\[
  \M=\{m_0,m_d,m_{\bar t}\}.
\]
Here, $m_0$ is a null content, $m_d$ certifies disability, and $m_{\bar t}$
certifies the absence of a disqualifying tenancy record. Their certifier sets are
\[
  M^{-1}(m_0)=\T,
  \qquad
  M^{-1}(m_d)=\{\{d\},\{d,t\}\},
  \qquad
  M^{-1}(m_{\bar t})=\{\emptyset,\{d\}\}.
\]
No other content is available.

Type $\{d\}$ can establish disability by submitting $m_d$ and can establish
clean history by submitting $m_{\bar t}$. Nevertheless,
$D^{\mathrm{clear}}$ is not content-certifiable. Accepting $m_d$ also admits
the disabled type with a disqualifying record, while accepting $m_{\bar t}$
also admits the nondisabled type with a clean record. Accepting both contents
does not solve the problem because they are alternative ways of passing the
test:
\[
  M^{-1}(m_d)\cup M^{-1}(m_{\bar t})
  =
  \{\emptyset,\{d\},\{d,t\}\},
\]
rather than the intersection
\[
  M^{-1}(m_d)\cap M^{-1}(m_{\bar t})
  =
  \{\{d\}\}.
\]
Thus, each component claim is separately certifiable, but the conjunction is
not. The missing object is a feasible content whose certifier set combines the
information in $m_d$ and $m_{\bar t}$.
\end{example}

The following elementary criterion converts the existence of a certificate
test into a type-by-type separation condition. It will be used repeatedly to
identify the certification burdens imposed by allocation rules.

\begin{lemma}[Certification criterion]\label{lem:cert}
For any $D\subseteq\T$, let $A_D^*:=\{m\in\M:M^{-1}(m)\subseteq D\}$
and its \emph{certifiable interior}
\[
  \mathcal I(D)
  :=
  \bigcup_{m\in A_D^*}M^{-1}(m).
\]
Then $A_D^*$ is the largest sound test for $D$, and the following are
equivalent:
\begin{enumerate}[label=(\roman*)]
  \item $D$ is content-certifiable;
  \item $A_D^*$ is complete for $D$;
  \item $\mathcal I(D)=D$.
\end{enumerate}
When these conditions hold, $A_D^*$ is the largest test that certifies $D$.
Moreover, $\mathcal I(D)$ is the largest content-certifiable subclass of $D$.
\end{lemma}

Lemma~\ref{lem:cert} provides the basic link between the evidence technology
and the claims required by an allocation rule. Rather than searching over all
possible certificate tests, it reduces certifiability to a type-by-type
separation condition: every intended claimant must possess some feasible
content that excludes all types outside the claimed class. The lemma is also
constructive. It identifies the maximal sound test $A^*_D$ and, when the target
class is not certifiable, the certifiable interior $\mathcal I(D)$, which
retains exactly the largest subclass that can be admitted without granting the
same claim to an unintended type. These objects provide the common foundation
for the results that follow: Theorems~\ref{thm:chain-burden}
and~\ref{thm:reserve-burden} apply the criterion to ranking upper contours and
reserve eligibility classes, respectively, while Theorem~\ref{thm:repairs}
uses $\mathcal I(D)$ to construct the canonical evidence-safe contraction of
reserve eligibility.


\subsection{Mimicry and Full Reports}

A transparent obstruction arises when one type can reproduce every content
available to another. We define the evidence-induced mimicry preorder.
\begin{definition}\label{def:mimicrypreo}
    Let $\preceqM$ be the evidence-induced mimicry preorder which has the following relation: 
    \[
      \theta\preceqM\theta'
      \quad\Longleftrightarrow\quad
      M(\theta)\subseteq M(\theta').
    \]   
\end{definition}
Thus, $\theta\preceqM\theta'$ means that $\theta'$ can mimic $\theta$: every
content feasible for $\theta$ is also feasible for $\theta'$. This preorder is
distinct from category inclusion. Category inclusion records which categories
a type has; mimicry records which contents it can produce. 

The mimicry relation gives an immediate observation of the necessary condition for certification. 
\begin{fact}
    Every content-certifiable class must be upward closed under $\preceqM$\footnote{Let $(X,\preceq)$ be a preordered set. A set $U \subseteq X$ is upward closed if, for every $u \in U$ and $x \in X$, $u \preceq x$ implies $x \in U$.}.
\end{fact}
Let $D$ be a content-certifiable set, if $\theta\in D$ and $\theta'$ can submit every content available to $\theta$, that is, $\theta \preceqM \theta'$, then by Definition~\ref{def:mimicrypreo}, $M(\theta) \subseteq M(\theta')$. Thus, any test passed by $\theta$ can also be passed by $\theta'$ as $\theta'$ can submit exactly the same content as $\theta$.

This condition is not sufficient. Upward closure rules out a single outsider
who can reproduce all of an intended claimant's evidence; certification
requires the claimant to possess a single content that every outsider fails to reproduce. Example~\ref{ex:separate-documents-no-conjunction} illustrates the difference. No outside type can reproduce both $m_d$ and $m_{\bar t}$, so no outside type mimics $\{d\}$. Yet neither content alone certifies $\{d\}$:
$m_d$ is also available to $\{d,t\}$, while $m_{\bar t}$ is also available to
$\emptyset$. Accepting both does not require the claimant to submit both,
because each accepted content is an alternative way to pass the test.
Certification would therefore require a single feasible content that is at
least as informative as the two contents together. The following condition
guarantees that such a content exists.

\begin{definition}[Full-report normality]\label{def:fullreport}
An evidence technology is \emph{full-report normal} if, for every
$\theta\in\T$ and every $m,m'\in M(\theta)$, there exists
$m''\in M(\theta)$ such that
\[
  M^{-1}(m'')
  \subseteq
  M^{-1}(m)\cap M^{-1}(m').
\]
\end{definition}

This is the finite analogue of the full-report condition of Lipman and Seppi
(1995) and evidentiary normality in Bull and Watson (2007). Since $M(\theta)$
is finite, repeated application of Definition~\ref{def:fullreport} yields a
feasible content $m^*(\theta)$ satisfying
\[
  M^{-1}\bigl(m^*(\theta)\bigr)
  =
  \bigcap_{m\in M(\theta)}M^{-1}(m)
  =
  \{\tau\in\T:\theta\preceqM\tau\}.
\]
Thus, a full report rules out every type that cannot reproduce the entire
evidence menu of $\theta$.

\begin{proposition}[Mimicry characterization under full reports]
\label{prop:fullreport}
Every content-certifiable class is upward closed under the mimicry preorder:
\[
  \theta\in D,\ \theta\preceqM\tau
  \quad\Longrightarrow\quad
  \tau\in D.
\]
If the evidence technology is full-report normal, the converse also holds.
Hence, under full-report normality,
\[
  D\text{ is content-certifiable}
  \quad\Longleftrightarrow\quad
  D\text{ is upward closed under }\preceqM.
\]
\end{proposition}

The first implication is universal: a certificate test cannot admit a type
while excluding another type that can reproduce all of its evidence. The
converse requires a full report. In Example~\ref{ex:separate-documents-no-conjunction},
the singleton $\{\{d\}\}$ is upward closed under mimicry, but it is not
content-certifiable because $m_d$ and $m_{\bar t}$ cannot be combined into one
jointly informative content. The example therefore violates full-report
normality.

\section{Certification Burdens of Allocation Rules}

We now turn to scarce allocation. Let $I=\{1,\ldots,n\}$ be a finite set of
applicants, and suppose the planner has
$q\in\{1,\ldots,n-1\}$ identical indivisible units. A feasible allocation gives
each applicant at most one unit and assigns at most $q$ units in total.

In this environment, the evidence required from applicants depends on how the
planner organizes priority. We study two broad classes of allocation rules. A
\emph{chain priority system} places applicants in a single ordered sequence of
priority classes; this includes highest-priority-category rules,
lexicographic rankings, additive scores, and other scalar rankings. A
\emph{reserve system} instead preserves separate eligibility labels and
assigns label-specific capacities. These classes do not exhaust all allocation
rules, but they generate distinct certification requirements.

Throughout this section, \emph{safe} means evidence-safe: a priority claim is
granted only when every type compatible with the submitted content satisfies
that claim.

\subsection{Chain Priority Systems}

We begin with chain priority systems, which place all applicant types on a
common ordered priority scale. Familiar examples include
highest-priority-criterion rules, lexicographic rules, additive scoring rules,
and other rules that produce a single ranking of applicant types. Although
these rules aggregate characteristics differently, each induces the same
reduced-form object: an ordered collection of priority classes.\footnote{The example allocation rules are defined formally in Appendix~\ref{app:chain-examples}}

\begin{definition}[Chain priority system]
A \emph{chain priority system} is generated by a target ranking represented by
an ordered partition
\[
  R=(R_1,\ldots,R_L)
\]
of $\T$, where $R_1$ is the highest-priority class and types within each class are tied. Let $r_R(\theta)=\ell$ if $\theta\in R_\ell$, and write
\[
  \theta'\succeq_R\theta
  \quad\Longleftrightarrow\quad
  r_R(\theta')\le r_R(\theta).
\]
\end{definition}

$\theta'\succeq_R\theta$ means that $\theta'$ has at least as much
priority as $\theta$. At every applicant profile, the system serves applicants in higher-priority classes before applicants in lower-priority classes. If capacity is exhausted within a class, applicants in that class are rationed equally.

The subsequent analysis takes $R$ as primitive. Once types have been assigned
to its ordered classes, the allocation stage uses only their positions on the
chain, not the underlying criteria used to construct the ranking.

\paragraph{The full-information chain allocation.}
For a type profile $\boldsymbol\theta = (\theta_1,\ldots,\theta_n)\in\T^n,$
let $N_\ell(\boldsymbol\theta) = \bigl|\{i:r_R(\theta_i)=\ell\}\bigr|$ denote the number of applicants in class $R_\ell$, and let
\[
  H_\ell(\boldsymbol\theta)
  =
  \sum_{k=1}^{\ell}N_k(\boldsymbol\theta),
  \qquad
  H_0(\boldsymbol\theta)=0.
\]
The cutoff class is $ c_q(\boldsymbol\theta) = \min\{\ell:H_\ell(\boldsymbol\theta)\ge q\}.$
The full-information chain allocation assigns applicant $i$ probability $x_i^{R,q}(\boldsymbol\theta)$ as follows:
\[
  x_i^{R,q}(\boldsymbol\theta)
  =
  \begin{cases}
    1,
      & r_R(\theta_i)<c_q(\boldsymbol\theta),\\[5pt]
    \displaystyle
    \frac{
      q-H_{c_q(\boldsymbol\theta)-1}(\boldsymbol\theta)
    }{
      N_{c_q(\boldsymbol\theta)}(\boldsymbol\theta)
    },
      & r_R(\theta_i)=c_q(\boldsymbol\theta),\\[12pt]
    0,
      & r_R(\theta_i)>c_q(\boldsymbol\theta).
  \end{cases}
\]
Thus, classes strictly above the cutoff are served, classes strictly below it
are not served, and only the cutoff class may be rationed. Once applicants
have been placed in the ordered classes, the allocation stage uses only their
positions on the chain; the underlying criteria are no longer used separately.

\paragraph{Dominant-strategy implementation under hard evidence.}

Applicants care only about their probability of receiving a unit, and
submitting evidence is costless. Let
\[
  X_q
  =
  \left\{
    x\in[0,1]^n:
    \sum_{i=1}^n x_i\le q
  \right\}
\]
denote the set of feasible allocation-probability vectors.

A \emph{content-based allocation mechanism} is a mapping $g:\M^n\to X_q.$
A feasible strategy for applicant $i$ is another mapping $s_i:\T\to\M$
satisfying $s_i(\theta)\in M(\theta)$ for every $\theta\in\T.$ A feasible strategy $s_i$ is weakly dominant under $g$ if, for every type
$\theta_i$, every alternative feasible content
$m_i\in M(\theta_i)$, and every profile of contents submitted by the other
applicants $m_{-i}\in\M^{n-1}$,
\[
  g_i\bigl(s_i(\theta_i),m_{-i}\bigr)
  \ge
  g_i(m_i,m_{-i}).
\]
Dominance is evaluated subject to the hard-evidence constraint: type
$\theta_i$ may deviate to any content in $M(\theta_i)$, but not to a content
outside that set.

The full-information allocation rule $x^{R,q}$ is
\emph{implementable in dominant strategies} if there exist a content-based
allocation mechanism $g$ and a feasible weakly dominant strategy $s_i$ for
each applicant $i$ such that $g_i \bigl(s_1(\theta_1),\ldots,s_n(\theta_n) \bigr) = x^{R,q}_i(\boldsymbol\theta)$ for every $i$ and $\boldsymbol\theta.$

\paragraph{Safe readings of evidence.}

We now connect the multi-applicant allocation problem to the one-person
certification problem. With partial evidence, the planner observes only submitted contents. Because the same content may be feasible for several types, the planner must interpret contents without assigning them more priority than they justify. Let us first define the upper contour set in this setting:
\begin{definition}[Upper Contour]\label{def:uppercon}
    For each class $\ell$, the \textbf{upper contour} of $\ell$ is
    \[
      U_\ell^R = \bigcup_{k\le\ell}R_k = \{\theta\in\T:r_R(\theta)\le\ell\}.
    \]    
\end{definition}

A \emph{reading} is a mapping $\kappa:\M\to\{1,\ldots,L\}$. Reading content $m$ as class $\ell$ amounts to treating it as evidence that
the submitter belongs to $U_\ell^R$. Since every type in $M^{-1}(m)$ can
submit $m$, such a reading is justified only if $M^{-1}(m)\subseteq U_\ell^R.$

Next, we define safe reading and implementability:
\begin{definition}
    A reading $\kappa$ is
    \begin{enumerate}[label=(\alph*)]
        \item \textbf{safe} for $R$ if $\kappa(m)\ge r_R(\theta)$ 
    for every $m \in \M$ and every $\theta \in M^{-1}(m)$. Equivalently, $M^{-1}(m)\subseteq U_{\kappa(m)}^R$ for every $m \in \M$.
        \item \textbf{implementable} for $R$ if every type can attain exactly its target class. That is,
        \[
          r_R(\theta) = \min_{m\in M(\theta)}\kappa(m) \quad\text{for every }\theta\in\T.
        \]
    \end{enumerate}
    We say that $R$ is \textbf{safely implementable} if there exists a reading that is safe and implements $R$.
\end{definition}

Safety prevents a content from being read more favorably than the
lowest-priority type that could have submitted it; implementability guarantees every type can attain exactly its target, $\min_{m\in M(\theta)}\kappa(m)$, which is $\theta$'s most favorable class available through its feasible contents. 

Among all safe readings, the canonical interpretation of content $m$ is its
\emph{skeptical class} $\bar r_R(m) := \max_{\tau\in M^{-1}(m)}r_R(\tau).$ It assigns the content the lowest priority held by any compatible type.

\begin{theorem}
\label{thm:chain-burden}
Fix a target ranking $R$. The following are equivalent:
\begin{enumerate}[label=(\roman*)]
  \item $R$ is safely implementable by a reading;

  \item every proper upper contour
  $U_\ell^R$, $\ell=1,\ldots,L-1$, is content-certifiable;

  \item every type $\theta\in\T$ has a feasible content $m\in M(\theta)$
  satisfying
  \[
    M^{-1}(m)\subseteq U_{r_R(\theta)}^R.
  \]
  Equivalently,
  \[
    r_R(\theta)
    =
    \min_{m\in M(\theta)}
    \max_{\tau\in M^{-1}(m)}r_R(\tau)
    \qquad\text{for every }\theta\in\T.
  \]
\end{enumerate}
Whenever these conditions hold, the skeptical reading $\bar r_R$ safely
implements $R$. These conditions are also equivalent to dominant-strategy
implementation of the full-information allocation rule
$x^{R,q}$ by a content-based allocation mechanism.
\end{theorem}

Theorem~\ref{thm:chain-burden} provides the formal bridge from one-person
certification to multi-person allocation. Against a profile at which one unit
remains, the allocation mechanism must distinguish types above a given
priority boundary from types below it. Any mechanism that reproduces the
full-information allocation therefore induces a certificate test for the
corresponding upper contour.

The result also identifies how much information a chain requires. Exact type
revelation is unnecessary. A type in $R_\ell$ needs only evidence that excludes
types below $R_\ell$; it may remain pooled with any type in a higher-priority
class. Thus, the direction of residual ambiguity matters more than the amount
of ambiguity. Upward pooling is harmless, whereas pooling across a lower
priority boundary can alter who receives a scarce unit.

If an upper contour is not content-certifiable, the failure has an allocation
consequence. On some applicant profiles, an intended claimant cannot be
guaranteed the marginal unit without making the same guarantee available to a
compatible lower-priority type. The planner must then either admit an
unsupported claimant or deny an intended claimant her target priority. The
dominant-strategy conclusion shows that allowing a more complicated mechanism
to condition on the entire profile of submitted contents cannot eliminate this
conflict.

\subsection{Reserve Systems}

Unlike a chain system, which places every applicant on a common priority
scale, a reserve system preserves separate eligibility labels and assigns
capacity to them. An applicant may qualify for several reserves, and
applicants with different labels need not be globally ranked against one
another.

Let $\Lcal$ be a finite set of active reserve labels. Each reserve
$r\in\Lcal$ has an integer capacity $q_r\geq 1$, and $q_0\geq 0$ denotes
open capacity, with
\[
  q_0+\sum_{r\in\Lcal}q_r=q.
\]
A target eligibility map $\rho:\T\to 2^\Lcal$ assigns each type its eligible reserve labels. For each $r\in\Lcal$, let $ D_r:=\{\theta\in\T:r\in\rho(\theta)\}$ denote the corresponding eligibility class.

\begin{definition}[Reserve system]
A \emph{reserve system} consists of the label set $\Lcal$, the capacity
vector $\mathbf q=(q_0,(q_r)_{r\in\Lcal}),$ the target eligibility map $\rho$, and a downstream reserve rule
\[
  \Phi^{\mathbf q}:(2^\Lcal)^n\to X_q.
\]
Moreover, $\Phi^q$ is \textbf{label-monotone} if, for every $i$, every profile
$S_{-i}\in(2^\Lcal)^{n-1}$ of the other applicants' certified label sets,
and every $S_i,S_i'\in2^\Lcal$ satisfying $S_i\subseteq S_i'$,
\[
  \Phi_i^{\mathbf q}(S_i,S_{-i})
  \leq
  \Phi_i^{\mathbf q}(S_i',S_{-i}).
\]
\end{definition}

At every profile of certified label sets, $\Phi^{\mathbf q}$ must be induced
by a feasible lottery over assignments: each applicant receives at most one
unit, at most $q_r$ units are assigned through reserve $r$ and only to
applicants certified for $r$, and at most $q_0$ open units are assigned.
Open units may be assigned to any applicant.

The rule specifies how applicants are rationed when demand for a reserve exceeds its capacity and how conflicts involving applicants with several labels are resolved. Throughout the analysis, we focus on label-monotone downstream reserve rules. Thus, holding the other applicants' certified labels fixed, adding labels to an applicant's certified set cannot reduce her allocation probability. Under full information, the resulting allocation rule is $x^{\rho,\mathbf q}(\boldsymbol\theta) := \Phi^{\mathbf q} \bigl(\rho(\theta_1),\ldots,\rho(\theta_n)\bigr).$

\paragraph{Label-specific applications.}

Applicants may submit evidence separately for different reserves. Let $\bot$
denote no application. A feasible application for type $\theta$ is a vector
\[
  a=(a_r)_{r\in\Lcal}
  \in(\M\cup\{\bot\})^\Lcal
\]
such that
\[
  a_r\in M(\theta)\cup\{\bot\}
  \qquad\text{for every }r\in\Lcal.
\]
Thus, an applicant may use different contents for different reserve labels.
The same content may also be submitted for more than one label.

A \emph{reserve reading} is a mapping $\eta:\M\to2^\Lcal,$ where $r\in\eta(m)$ means that content $m$ is accepted as evidence of
eligibility for reserve $r$. Set $\eta(\bot)=\emptyset$. Under application
$a$, label $r$ is certified precisely when $r\in\eta(a_r).$ Similar to the chain priority systems, we define safe reading and implementability for reserve systems.
\begin{definition}
    A reading $\eta$ is
     \begin{enumerate}[label=(\alph*)]
        \item \textbf{safe} for $\rho$ if it assigns a label only when every type capable of submitting the content holds that label: 
        \[
          \eta(m) \subseteq \bigcap_{\theta\in M^{-1}(m)}\rho(\theta)
          \qquad\text{for every }m\in\M.
        \]
        Equivalently,
        \[
          r\in\eta(m) \quad\Longrightarrow\quad
          M^{-1}(m)\subseteq D_r.
        \]
        \item \textbf{implementable} for $\rho$ if
        \[
          \bigcup_{m\in M(\theta)}\eta(m) =
          \rho(\theta) \qquad\text{for every }\theta\in\T.
        \]
    \end{enumerate}
    We say that $\rho$ is \emph{safely implementable} if there exists a safe reserve reading that implements $\rho$.
\end{definition}

Safety guarantees that the content is read in a way such that it is contained in the intersection of $\rho(\theta)$; implementation requires every target label to be attainable. Given a reserve reading, the downstream rule is applied to the certified label sets generated by applicants' application vectors. A reserve strategy specifies a feasible application vector for each type, and weak dominance is defined as in the chain system.

\paragraph{Skeptical reserve reading.}

The evidence technology provides a canonical safe interpretation of every
content. The \emph{skeptical reserve reading} is defined by $\bar\eta_\rho(m) := \bigcap_{\theta\in M^{-1}(m)}\rho(\theta).$
Thus, $\bar\eta_\rho(m)$ contains exactly the labels held by every type
compatible with $m$. A reserve reading $\eta$ is safe if and only if
\[
  \eta(m)\subseteq\bar\eta_\rho(m)
  \qquad\text{for every }m\in\M.
\]
Hence, $\bar\eta_\rho$ is the pointwise largest safe reserve reading.

For any eligibility map $\rho$, define the \emph{verifiable-label operator}
$\mathcal V$ by
\[
  (\mathcal V\rho)(\theta)
  :=
  \bigcup_{m\in M(\theta)}\bar\eta_\rho(m).
\]
The set $(\mathcal V\rho)(\theta)$ consists of all labels that type $\theta$
can safely establish under label-specific applications. Since
$\theta\in M^{-1}(m)$ for every $m\in M(\theta)$,
\[
  (\mathcal V\rho)(\theta)\subseteq\rho(\theta).
\]
Thus, the fixed-point condition $\mathcal V\rho=\rho$ requires every target
label, and only a target label, to be safely attainable.

\begin{theorem}
\label{thm:reserve-burden}
Fix a target eligibility map $\rho:\T\to2^\Lcal$. The following are equivalent:
\begin{enumerate}[label=(\roman*)]
  \item $\rho$ is safely implementable;

  \item every eligibility class $D_r$, $r\in\Lcal$, is
  content-certifiable;

  \item for every $r\in\Lcal$ and every $\theta\in D_r$, there exists
  $m\in M(\theta)$ such that
  \[
    M^{-1}(m)\subseteq D_r.
  \]
  Equivalently,
  \[
    \mathcal V\rho=\rho.
  \]
\end{enumerate}
Whenever these conditions hold, the skeptical reserve reading
$\bar\eta_\rho$ implements $\rho$. Moreover, under every feasible,
label-monotone downstream reserve rule, applications that certify all labels
in $\rho(\theta)$ are weakly dominant and implement the full-information
allocation $x^{\rho,\mathbf q}$.

If the institution instead requires one common content to establish an
applicant's entire reserve-label set, then a safe reserve reading implements
$\rho$ if and only if
\[
  \text{for every }\theta\in\T,\qquad
  \exists m\in M(\theta)
  \text{ such that }
  \bar\eta_\rho(m)=\rho(\theta).
\]
Whenever this condition holds, the skeptical reserve reading implements
$\rho$ under the common-content requirement.
\end{theorem}

Theorem~\ref{thm:reserve-burden} identifies the proof obligation generated by
each reserve. Under label-specific applications, access to reserve $r$
requires a content that excludes every type outside $D_r$. The content need
not reveal the applicant's complete type, distinguish among types that are all
eligible for $r$, or establish eligibility for any other reserve. It must
establish only the capacity-relevant claim $\theta\in D_r$.

When reserve $r$ binds, this condition determines who may compete for its
scarce capacity. Accepting a content also available to an ineligible type
admits an unsupported claimant, whereas rejecting it may exclude an intended
claimant. Condition~(iii) eliminates this conflict by ensuring that every
eligible type has some qualifying content and that no ineligible type can use
that content to obtain the label.

The theorem characterizes implementation of the eligibility map rather than
reproduction of one particular allocation outcome. A label may be redundant
under a particular capacity vector or on a profile at which its reserve is
undersubscribed. Implementing the eligibility map is the rule-independent
requirement that guarantees correct implementation across profiles and for
every label-monotone reserve procedure using $\rho$.

\paragraph{Comparing Theorem~\ref{thm:chain-burden} and
Theorem~\ref{thm:reserve-burden}.}

The two theorems connect one-person certification to the allocation of scarce
units, but they generate different proof obligations. A chain places each
applicant on one ordered scale. Since a type in class $R_\ell$ may safely
remain pooled with better types, the relevant classes are the upper contours $U_\ell^R.$

A reserve system instead uses each label as a separate gate to a portion of
capacity. The relevant classes are therefore the label-specific eligibility
sets $D_r.$

The strategic choices differ accordingly. Under a chain, an applicant chooses
one feasible content yielding her most favorable safe class. Under
label-specific reserves, she may choose a different content for each label and
applies for every label she can establish. A content crossing a chain boundary
may allow a lower-priority type to compete for the marginal unit; a content
crossing the boundary of $D_r$ may allow an ineligible type to enter reserve
$r$. The same evidence technology can therefore support one allocation system
while failing the certification burden imposed by the other.

Under full-report normality, safe implementability in both architectures is
equivalent to monotonicity with respect to evidence mimicry. Recall that
$\theta\preceqM\tau$ means that $\tau$ can submit every content available to
$\theta$. Hence,
\[
  \begin{aligned}
  R\text{ is safely implementable}
  &\quad\Longleftrightarrow\quad
  \forall\theta,\tau\in\T,\quad
  \theta\preceqM\tau
  \ \Longrightarrow\
  \tau\succeq_R\theta,\\
  \rho\text{ is safely implementable}
  &\quad\Longleftrightarrow\quad
  \forall\theta,\tau\in\T,\quad
  \theta\preceqM\tau
  \ \Longrightarrow\
  \rho(\theta)\subseteq\rho(\tau).
  \end{aligned}
\]
In words, a type that can mimic another type cannot be assigned lower
priority in a chain or fewer eligibility labels in a reserve system. Under
full-report normality, this monotonicity condition is both necessary and
sufficient; without full reports, it is only necessary.

\subsection{Evidence-Safe Repairs}

Theorems~\ref{thm:chain-burden} and~\ref{thm:reserve-burden}
characterize when a target rule can be implemented using the available
evidence. When their conditions fail, the planner may enrich the evidence
technology, for example by adding a jointly informative content, or weaken the
target rule so that every retained priority claim is supported by the existing
technology. We study the second response.

The repairs considered here are \emph{conservative}. A chain repair may demote
a type but never place it above its target class. A reserve repair may delete a
target eligibility label but never add a label absent from the target rule.
The resulting notion of minimality is pointwise and order-theoretic. 

\paragraph{Chain repairs.}

Write the target chain as a class-index function $r^0:\T\to\{1,\ldots,L\},$
where lower indices mean higher priority. For any class-index function $r$,
define its \emph{skeptical closure} by
\[
  (\mathcal S r)(\theta)
  =
  \min_{m\in M(\theta)}
  \max_{\tau\in M^{-1}(m)}r(\tau).
\]
The inner maximum is the most favorable safe reading of content $m$, and the
outer minimum is the best such reading attainable by type $\theta$. Hence,
$(\mathcal S r)(\theta)$ is the best class that $\theta$ can safely establish
under ranking $r$.

Because $\theta\in M^{-1}(m)$ whenever $m\in M(\theta)$,
\[
  (\mathcal S r)(\theta)\ge r(\theta)
  \quad\text{for every }\theta\in\T.
\]
Thus, $\mathcal S$ can only demote types. By the skeptical-reading
characterization in Theorem~\ref{thm:chain-burden}, a class-index function
$r$ is safely implementable if and only if
\[
  \mathcal S r=r.
\]

A class-index function $r$ is a \emph{relaxation} of $r^0$ if $r(\theta)\ge r^0(\theta)$ for every $\theta \in \T$.
A relaxation may demote types, but it never grants a type a better absolute
class than the target rule assigns.
We do not require a repaired class-index function to be onto; empty classes
may be deleted without changing the induced ranking or allocation. Minimality
is evaluated on the original class-index scale, before deleting empty classes.
This criterion does not require preserving pairwise rankings or original ties
and does not measure allocation or welfare losses.

\paragraph{Reserve repairs.}

Recall from Lemma~\ref{lem:cert} that $\mathcal I(D)$ is the certifiable
interior of $D$: the largest content-certifiable subclass of $D$. In
particular,
\[
  D\text{ is content-certifiable}
  \quad\Longleftrightarrow\quad
  \mathcal I(D)=D.
\]

Let $\rho^0:\T\to2^\Lcal$ be a target reserve eligibility mapping, so $\rho^0(\theta)$ is the set of labels
for which the planner intends type $\theta$ to be eligible. The same
information can be read label by label. For each $r\in\Lcal$, $D_r^0 = \{\theta\in\T:r\in\rho^0(\theta)\}$
is the \emph{target eligibility class} of label $r$: the set of types the
planner intends to admit to reserve $r$. The mapping $\rho^0$ and the family
$(D_r^0)_{r\in\Lcal}$ describe the same object from opposite sides. The mapping lists the labels of a type; the class lists the types of a label. We repair on the class side, because content-certifiability is a property of type classes, and then read the result back as a mapping.

Under the label-specific application convention, a reserve mapping $\rho$ is
\emph{evidence-safe} if it is implementable by a safe reserve reading, which by
Theorem~\ref{thm:reserve-burden} holds exactly when every class
$\{\theta\in\T:r\in\rho(\theta)\}$ is content-certifiable. It is a
\emph{contraction} of $\rho^0$ if
\[
  \rho(\theta)\subseteq\rho^0(\theta)
  \quad\text{for every }\theta\in\T,
\]
that is, if it removes labels but never adds them. The following theorem gives the canonical conservative repair for each system.

Define the evidence-safe repairs of the target chain and reserve system by
\[
  r^E:=\mathcal S r^0,
  \qquad
  D_r^E:=\mathcal I(D_r^0)
  \quad\text{for each }r\in\Lcal, \quad \textrm{and} \quad
  \rho^E:=\mathcal V\rho^0.
\]

\begin{theorem}
\label{thm:repairs}
Fix a finite evidence technology $M(\cdot)$.
\begin{enumerate}[label=(\arabic*)]
  \item \emph{Chain priority systems.}
  The skeptical operator $\mathcal S$ is a closure operator: it is monotone,
  extensive, and idempotent. Consequently,
  \[
    r^E=\mathcal S r^0
  \]
  is the pointwise smallest safely implementable relaxation of the target
  class-index function $r^0$.

  \item \emph{Reserve systems.}
  The certifiable-interior operator $\mathcal I$ is monotone, contractive, and idempotent. Repair each label separately, $D_r^E=\mathcal I(D_r^0)$ for every  $r \in \Lcal$, and define $\rho^E(\theta)=\{r\in\Lcal:\theta\in D_r^E\}.$
  
  Then $D_r^E$ is the largest content-certifiable subclass of $D_r^0$ for every $r$, and $\rho^E$ is the pointwise largest evidence-safe contraction of $\rho^0$.
\end{enumerate}
Both repaired rules are implemented by their skeptical readings. Empty repaired
chain classes may be deleted without changing the induced ranking or
allocation.
\end{theorem}

The theorem gives closed-form, one-step repairs:
\[
  r^E(\theta)
  =
  \min_{m\in M(\theta)}
  \max_{\tau\in M^{-1}(m)}r^0(\tau),
  \qquad
  D_r^E
  =
  \bigcup_{m:\,M^{-1}(m)\subseteq D_r^0}M^{-1}(m).
\]
The chain repair makes the smallest pointwise demotions needed for safe
implementation, while the reserve repair retains, label by label, the largest
content-certifiable subset of the target eligibility class. Because
$\mathcal S$ and $\mathcal I$ are idempotent, the repaired objects are fixed
points: applying either repair again produces no further change.

The two repairs use the same certifiable-interior operation. Indeed, for every
cutoff $\ell\in\{1,\ldots,L\}$,
\[
  \{\theta\in\T:(\mathcal S r^0)(\theta)\le\ell\}
  =
  \mathcal I\bigl(\{\theta\in\T:r^0(\theta)\le\ell\}\bigr).
\]
Both sides consist of types possessing a feasible content whose certifier set
lies within the target upper contour. Thus, chains apply $\mathcal I$ to
upper contours, whereas reserves apply it to eligibility classes.

These fixed points are directly implementable by their skeptical readings.
Under the repaired chain, every type has a feasible content attaining its
repaired class; submitting such a content is weakly dominant under the chain
cutoff rule and implements the corresponding full-information allocation for
every interior capacity. Under the repaired reserve map $\rho^E$, the
skeptical reading $\bar\eta_{\rho^E}$ allows each type to establish exactly
the labels in $\rho^E(\theta)$. With label-specific applications, applying
for all such labels is weakly dominant under every feasible label-monotone
downstream reserve rule.

The repairs are independent of the realized applicant profile and of which
capacity ultimately binds. They may leave the realized allocation unchanged
when the affected chain boundary or reserve is slack. When it binds, however,
a demotion or label deletion can change who receives a unit or who is subject
to rationing. The repairs therefore remove unsupported priority claims while
preserving as much of the target rule as the evidence technology permits.

Also, note that a repair does not recover an infeasible target allocation. It replaces
the target rule with the pointwise closest conservative revision: unsupported
chain priorities are relaxed by demotion, while unsupported reserve labels are
deleted.

\begin{example}[Repair under uninformative evidence]
\label{ex:repair-binding-unit}
Let $\T=\{h,\ell\}$, and suppose both types can submit only the same null
content:
\[
  M^{-1}(m_0)=\{h,\ell\}.
\]
The evidence technology therefore cannot support any distinction between the
two types.

Suppose first that the target chain ranks $h$ above $\ell$. The skeptical
repair pools the two types into one priority class. With one unit and one
applicant of each type, the repaired chain rations the unit equally. This
implements the repaired ranking, although the original priority distinction
cannot be preserved.

Suppose instead that there is one reserve label $r$ with target eligibility
class $D_r^0=\{h\}$. Because the only content is also available to $\ell$,
\[
  D_r^E=\mathcal I(D_r^0)=\emptyset.
\]
The repaired eligibility map therefore certifies neither type for $r$ and is
implemented by accepting no content for that label. Under the present
feasibility convention, the capacity assigned to this reserve remains unused.

The example illustrates what a conservative repair does. It does not recover
an infeasible target. It removes only the unsupported claims: the chain repair
eliminates an unsupported ranking distinction by pooling, whereas the reserve
repair eliminates an unsupported eligibility label.
\end{example}

\section{Applications}

The preceding analysis treats evidence abstractly through the set of types
capable of submitting each content. We now specialize to \emph{category
evidence}, in which certificates make verifiable statements about the
categories constituting an applicant's type. This structure distinguishes
evidence of presence from evidence of absence and makes clear when several
claims can be established jointly.

\subsection{Signed-Bundle Evidence}

A signed bundle combines positive and negative category claims in a single
certificate.

\begin{definition}[Signed-bundle evidence]
A \emph{signed-bundle certificate} is a pair
\[
  (P,A),
  \qquad
  P,A\subseteq\C,
  \qquad
  P\cap A=\emptyset,
\]
where the categories in $P$ are certified present and those in $A$ are
certified absent.

A \emph{signed-bundle evidence technology} is a collection $G$ of such
certificates. The certificates feasible for type $\theta$ are
\[
  G(\theta)
  :=
  \{(P,A)\in G:
    P\subseteq\theta,\;
    A\cap\theta=\emptyset\}.
\]
We assume $G(\theta)\neq\emptyset$ for every $\theta\in\T$.
\end{definition}

For a certificate $m=(P,A)$, the corresponding certifier set is
\[
  M^{-1}(m)
  =
  \T(P,A)
  :=
  \{\tau\in\T:
    P\subseteq\tau,\;
    A\cap\tau=\emptyset\}.
\]
Thus, $\T(P,A)$ consists of exactly those types for which every presence and
absence claim contained in the certificate is true.

The general characterizations now specialize directly to signed bundles.

\begin{corollary}[Implementation under signed-bundle evidence]
\label{cor:signed}
Under a signed-bundle evidence technology $G$, the following statements hold.
\begin{enumerate}[label=(\roman*)]
  \item A target chain ranking $R$ is safely implementable if and only if,
  for every $\theta\in\T$, there exists $(P,A)\in G(\theta)$ such that
  \[
    \T(P,A)\subseteq U_{r_R(\theta)}^R.
  \]

  \item Under label-specific reserve applications, a target eligibility map
  $\rho:\T\to2^\Lcal$ is safely implementable if and only if, for every
  $r\in\Lcal$ and every $\theta\in D_r$, there exists
  $(P,A)\in G(\theta)$ such that
  \[
    \T(P,A)\subseteq D_r.
  \]
\end{enumerate}
\end{corollary}

By Theorems~\ref{thm:chain-burden} and~\ref{thm:reserve-burden}, these same
conditions yield dominant-strategy implementation of the corresponding
full-information allocations: for chains at every interior capacity, and for
reserves under every feasible label-monotone downstream rule.

The two architectures impose different requirements on the same certificate.
For a chain, the bundle need only exclude types below the applicant's target
priority boundary. For reserve $r$, it must exclude every type outside the
eligibility class $D_r$. Under label-specific applications, an applicant may
use different bundles to establish different reserve labels. Within either
architecture, however, a claim involving both presence and absence requires a
single feasible bundle that carries those statements jointly; separately
available certificates do not by themselves establish their conjunction.

\subsection{One-Sided and Exact Evidence}

Signed-bundle evidence contains three useful benchmark technologies.

\paragraph{Positive evidence.}
Positive evidence permits only bundles of the form $(P,\emptyset)$. Such a
bundle can be submitted by exactly the types containing every category in
$P$:
\[
  \T(P,\emptyset)
  =
  \{\theta\in\T:P\subseteq\theta\}.
\]
A positive evidence technology is \emph{rich} if it contains
$(P,\emptyset)$ for every $P\subseteq\C$.

\paragraph{Negative evidence.}
Negative evidence permits only bundles of the form $(\emptyset,A)$. Such a
bundle can be submitted by exactly the types containing none of the categories
in $A$:
\[
  \T(\emptyset,A)
  =
  \{\theta\in\T:A\cap\theta=\emptyset\}.
\]
A negative evidence technology is \emph{rich} if it contains
$(\emptyset,A)$ for every $A\subseteq\C$.

\paragraph{Exact-profile evidence.}
Exact-profile evidence contains, for every type $\theta$, the bundle
\[
  (\theta,\C\setminus\theta).
\]
Its certifier set is the singleton $\{\theta\}$, so the applicant can establish
her complete category profile.

The relevant restrictions on target rules are defined with respect to category
inclusion.

\begin{definition}[Inclusion monotonicity]
A chain ranking $R$ is \emph{inclusion-monotone} if
\[
  \theta\subseteq\theta'
  \quad\Longrightarrow\quad
  \theta'\succeq_R\theta,
\]
and \emph{inclusion-antitone} if
\[
  \theta\subseteq\theta'
  \quad\Longrightarrow\quad
  \theta\succeq_R\theta'.
\]

A reserve eligibility map $\rho$ is \emph{inclusion-monotone} if
\[
  \theta\subseteq\theta'
  \quad\Longrightarrow\quad
  \rho(\theta)\subseteq\rho(\theta'),
\]
and \emph{inclusion-antitone} if
\[
  \theta\subseteq\theta'
  \quad\Longrightarrow\quad
  \rho(\theta')\subseteq\rho(\theta).
\]
\end{definition}

\begin{proposition}[Rich one-sided and exact evidence]
\label{prop:onesided}
The following statements hold.
\begin{enumerate}[label=(\arabic*)]
  \item \emph{Chain priority systems.}
  \begin{enumerate}[label=(\roman*),leftmargin=*]
    \item Under rich positive evidence, $R$ is safely implementable if and
    only if it is inclusion-monotone.
    \item Under rich negative evidence, $R$ is safely implementable if and
    only if it is inclusion-antitone.
    \item Under exact-profile evidence, every target ranking is safely
    implementable.
  \end{enumerate}

  \item \emph{Reserve systems.} Under label-specific reserve applications,
  \begin{enumerate}[label=(\roman*),leftmargin=*]
    \item under rich positive evidence, safe implementability of $\rho$ is
    equivalent to inclusion monotonicity;
    \item under rich negative evidence, safe implementability of $\rho$ is
    equivalent to inclusion antitonicity;
    \item under exact-profile evidence, every target eligibility map is safely
    implementable.
  \end{enumerate}
\end{enumerate}
\end{proposition}

The direction of these conditions follows from the information contained in
one-sided evidence. A positive bundle feasible for $\theta$ is also feasible
for every type containing $\theta$. Positive evidence can therefore support a
priority or eligibility claim only if that claim is preserved when categories
are added. A negative bundle feasible for $\theta$ is also feasible for every
type contained in $\theta$, so negative evidence can support only claims
preserved when categories are removed. Exact-profile evidence imposes no such
directional restriction because each type can distinguish itself from every
other type.

Richness is required for sufficiency, but not for necessity. Under any
positive technology, safe implementability requires inclusion monotonicity;
under any negative technology, it requires inclusion antitonicity. Richness
ensures that each type has the maximally informative one-sided certificate:
$(\theta,\emptyset)$ under positive evidence and
$(\emptyset,\C\setminus\theta)$ under negative evidence. Without these
certificates, a target rule may satisfy the relevant monotonicity condition
but remain unsupported by the available evidence.

Rich positive, rich negative, and exact-profile evidence all satisfy
full-report normality. Proposition~\ref{prop:onesided} can therefore also be
viewed as a category-specific application of
Proposition~\ref{prop:fullreport} and
Theorems~\ref{thm:chain-burden}--\ref{thm:reserve-burden}.

\subsection{When Reserve Eligibility Can Be Represented by a Chain}

Certification and representation are distinct requirements. Even when every
reserve eligibility class is content-certifiable, a chain may be unable to
represent those classes through priority cutoffs. The reason is that the
upper contours of a chain are necessarily nested, whereas reserve eligibility
classes may cross. The following proposition characterizes exactly when all
nonempty reserve classes can be represented as upper contours of a single chain.

\begin{proposition}[Chain representability of reserve eligibility]
\label{prop:nestedness}
Let
\[
  \mathcal D(\rho)
  :=
  \{D_r:r\in\Lcal,\ D_r\neq\emptyset\}
\]
denote the collection of nonempty reserve eligibility classes, with duplicate
classes identified. There exists a chain ranking $R$ on $\T$ such that every
class in $\mathcal D(\rho)$ is an upper contour of $R$ if and only if
$\mathcal D(\rho)$ is nested by inclusion: for every
$D,D'\in\mathcal D(\rho)$,
\[
  D\subseteq D'
  \quad\text{or}\quad
  D'\subseteq D.
\]
\end{proposition}

Necessity follows because any two upper contours of a common chain are nested.
For sufficiency, order the distinct nonempty eligibility classes as
\[
  D^1\subsetneq D^2\subsetneq\cdots\subsetneq D^K.
\]
The successive differences
\[
  D^1,\quad
  D^2\setminus D^1,\quad\ldots,\quad
  D^K\setminus D^{K-1},
\]
together with $\T\setminus D^K$ when it is nonempty, form a chain whose first
$k$ classes have union $D^k$. Thus, nested nonempty reserve classes can be
represented as priority cutoffs of a scalar rule. Cross-cutting classes cannot
be represented simultaneously as upper contours of any score, lexicographic
rule, or other chain ranking.

Proposition~\ref{prop:nestedness} concerns representation of eligibility
classes by upper contours, not equivalence of allocation procedures. For a
fixed chain, a nonempty reserve class that is not an upper contour yields an
allocation disagreement in the corresponding single-reserve contest.

\begin{proposition}[Chain--reserve divergence]
\label{prop:divergence}
Fix a target chain ranking $R$ and a reserve eligibility map $\rho$. Suppose
some nonempty eligibility class $D_r$ is not an upper contour of $R$.
Equivalently, suppose there exist $\theta,\tau\in\T$ such that
\[
  \theta\in D_r,
  \qquad
  \tau\notin D_r,
  \qquad
  \tau\succeq_R\theta.
\]
Then there is a two-applicant, one-unit profile on which the full-information
chain allocation differs from the allocation with one unit reserved for label
$r$ and no open capacity.
\end{proposition}

The disagreement in Proposition~\ref{prop:divergence} is representational,
not evidentiary. On the profile containing $\theta$ and $\tau$, feasibility
requires the reserve to give the ineligible type $\tau$ probability zero.
The chain gives $\tau$ probability one when she has strictly higher priority
and one-half when the types are tied. This difference can arise even when
the ranking and every reserve label are correctly certified.

Together, Propositions~\ref{prop:nestedness} and~\ref{prop:divergence} show
that cross-cutting reserve classes cannot all be upper contours of one scalar
ranking. For every proposed chain, at least one reserve class generates a
single-reserve contest whose allocation differs from the chain.

\subsection{A Stylized Housing Application}

This example separates two constraints: evidence must establish the claims
used by an allocation rule, and the rule must preserve the distinctions
needed for its intended priorities. Better evidence can resolve the first
constraint without resolving the second.

Let $\C=\C^+\cup\C^-$, where $\C^+=\{V,D,H,F\}$ denotes veteran status,
disability, homelessness, and family status, and $\C^-=\{T\}$ denotes an
adverse tenancy record. Types are $\theta\in\T=2^\C$. Consider the score
\[
  s(\theta)
  =\mathbf 1\{V\in\theta\}
   +2\mathbf 1\{D\in\theta\}
   +2\mathbf 1\{H\in\theta\}
   +\mathbf 1\{F\in\theta\}
   +2\mathbf 1\{T\notin\theta\},
\]
with higher scores receiving higher priority.

\emph{The evidence constraint.}
Positive evidence cannot establish clearance, while negative evidence cannot
establish favorable status. The score is neither inclusion-monotone nor
inclusion-antitone, so neither one-sided technology safely implements it
(Proposition~\ref{prop:onesided}). Suppose instead that every type can submit
\[
  m(\theta)
  :=\left(\theta\cap\C^+,\, 
          \{c\in\C^-:c\notin\theta\}\right).
\]
The first component certifies favorable characteristics that are present;
the second certifies adverse characteristics that are absent. Every type
compatible with this bundle has score at least $s(\theta)$, and $\theta$
itself is compatible. Its skeptical score is therefore exactly $s(\theta)$.
These claims must be jointly readable, possibly through an authenticated
portfolio of documents.

\emph{The representation constraint.}
Now reserve one unit for disability with clearance and one for homelessness
with clearance, with no open capacity. The eligibility classes are
\[
  D_{D\text{-res}}
  =\{\theta\in\T:D\in\theta,\ T\notin\theta\},
  \qquad
  D_{H\text{-res}}
  =\{\theta\in\T:H\in\theta,\ T\notin\theta\}.
\]
The signed bundles certify both classes, but neither class contains the
other. Proposition~\ref{prop:nestedness} therefore rules out representing
both as upper contours of a single chain.

For a simple allocation comparison, take three applicants with types
$\{D\}$, $\{V,F\}$, and $\{H\}$, all submitting their signed bundles.
Their scores are all $4$, so each receives a unit with probability $2/3$.
A reserve rule that fills each reserve whenever feasible instead assigns
its units to the disability and homelessness applicants. The difference
persists despite adequate evidence: the score treats different sources of
priority as interchangeable, whereas the reserves retain separate labels.

A different score could reproduce this particular allocation, but no score
can make both crossing eligibility classes upper contours.
\emph{Improving evidence can establish the relevant claims; it cannot make
a scalar ranking represent crossing reserve-eligibility classes as upper contours.}

\section{Conclusion}

This paper studies which allocation rules can be implemented when applicants
have private characteristics and can submit hard but partial evidence. The
central friction is that hard evidence verifies what is submitted without
necessarily revealing whether other relevant evidence has been withheld. We
show that a chain priority system is safely implementable exactly when every
proper upper contour of its ranking is content-certifiable. Under
label-specific applications, a reserve eligibility map is safely implementable
exactly when each reserve eligibility class is content-certifiable. In both
architectures, skeptical readings provide canonical implementations and
support dominant-strategy implementation of the corresponding full-information
allocations, with the reserve conclusion requiring a feasible, label-monotone
downstream rule.

When these conditions fail, the available evidence cannot support the original
target rule. We therefore construct canonical one-step repairs. A chain is
repaired through the smallest pointwise priority demotions, whereas a reserve
system is repaired by deleting unsupported labels while retaining the largest
certifiable part of each eligibility class. These architecture-specific
repairs reflect a deeper distinction. A chain removes unsupported priority
claims through pointwise demotion, while a reserve system removes
unsupported access to label-specific capacity. Moreover, nonempty reserve
eligibility classes can be represented as upper contours of a single chain
exactly when they are nested. Cross-cutting classes cannot all be priority
cutoffs of any scalar ranking, even when the underlying eligibility claims
are fully certifiable.

The broader implication is that allocation design and evidence design cannot
be separated. An institution must determine not only which characteristics
should confer priority, but also which type classes applicants must establish
and whether the available evidence can certify those classes, including claims
that combine favorable characteristics with the absence of adverse conditions.
The framework provides a direct procedure: identify the claims selected by the
allocation rule, test their certifiability, check nestedness when
representation by priority cutoffs is desired, and apply the canonical repair when the target
cannot be supported. Natural extensions would endogenize evidence acquisition
and certificate design and connect these proof burdens to axiomatic
foundations for chain and reserve procedures.

\newpage

\begin{appendices}

\section{Proofs}

\begin{proof}[Proof of Lemma~\ref{lem:cert}]
If a test $A$ is sound for $D$, then every $m\in A$ satisfies
$M^{-1}(m)\subseteq D$. Hence $A\subseteq A_D^*$, so $A_D^*$ is the largest
sound test for $D$. Its passing class is
\[
  \{\theta\in\T:M(\theta)\cap A_D^*\neq\emptyset\}
  =
  \bigcup_{m\in A_D^*}M^{-1}(m)
  =
  \mathcal I(D).
\]
Therefore, $A_D^*$ is complete for $D$ if and only if
$\mathcal I(D)=D$.

If some test $A$ certifies $D$, then $A\subseteq A_D^*$ by soundness, while
completeness of $A$ implies that every $\theta\in D$ passes $A_D^*$. Hence
$A_D^*$ is complete. Conversely, if $A_D^*$ is complete, then it is both sound
and complete and therefore certifies $D$. This proves the three equivalent
conditions and the maximality of $A_D^*$ among certifying tests.

Finally, $A_D^*$ certifies its own passing class $\mathcal I(D)$, so
$\mathcal I(D)$ is content-certifiable. If $\widetilde D\subseteq D$ is any
content-certifiable class, the union characterization just proved gives
\[
  \widetilde D
  =
  \bigcup_{m:\,M^{-1}(m)\subseteq\widetilde D}M^{-1}(m)
  \subseteq
  \bigcup_{m:\,M^{-1}(m)\subseteq D}M^{-1}(m)
  =
  \mathcal I(D).
\]
Thus, $\mathcal I(D)$ is the largest content-certifiable subclass of $D$.
\end{proof}

\begin{proof}[Proof of Proposition~\ref{prop:fullreport}]
Suppose $D$ is content-certifiable, $\theta\in D$, and
$\theta\preceqM\tau$. Let $A$ certify $D$. Completeness gives some
$m\in M(\theta)\cap A$. Since $M(\theta)\subseteq M(\tau)$, the same content is
feasible for $\tau$, so $\tau$ also passes $A$ and belongs to $D$. Hence every
content-certifiable class is upward closed under mimicry.

Now suppose the technology is full-report normal and $D$ is upward closed.
For each $\theta\in D$, finiteness and repeated application of full-report
normality yield $m^*(\theta)\in M(\theta)$ such that
\[
  M^{-1}\bigl(m^*(\theta)\bigr)
  \subseteq
  \bigcap_{m\in M(\theta)}M^{-1}(m).
\]
The reverse inclusion holds because $m^*(\theta)\in M(\theta)$. Therefore,
\[
  M^{-1}\bigl(m^*(\theta)\bigr)
  =
  \{\tau\in\T:M(\theta)\subseteq M(\tau)\}
  =
  \{\tau\in\T:\theta\preceqM\tau\}
  \subseteq D,
\]
where the final inclusion uses upward closure. Every $\theta\in D$ thus has a
feasible content whose certifier set is contained in $D$. Lemma~\ref{lem:cert}
then implies that $D$ is content-certifiable.
\end{proof}

\begin{proof}[Proof of Theorem~\ref{thm:chain-burden}]
For any reading $\kappa$, safety is equivalent to
\[
  \kappa(m)
  \ge
  \max_{\tau\in M^{-1}(m)}r_R(\tau)
  =
  \bar r_R(m)
  \qquad\text{for every }m\in\M.
\]
Hence, $\bar r_R$ is safe and is the pointwise most favorable safe reading.

We first prove the equivalence of conditions~(i) and~(iii). Suppose that a
safe reading $\kappa$ implements $R$. For each $\theta$, choose
$m\in M(\theta)$ attaining the minimum in the implementation condition, so
$\kappa(m)=r_R(\theta)$. Safety implies
$\bar r_R(m)\le r_R(\theta)$, while $\theta\in M^{-1}(m)$ implies
$\bar r_R(m)\ge r_R(\theta)$. Thus,
\[
  \bar r_R(m)=r_R(\theta),
\]
which is equivalent to
\[
  M^{-1}(m)\subseteq U_{r_R(\theta)}^R.
\]
This proves condition~(iii).

Conversely, suppose condition~(iii) holds. For every $m\in M(\theta)$,
$\theta\in M^{-1}(m)$ gives $\bar r_R(m)\ge r_R(\theta)$. The content supplied
by condition~(iii) gives the reverse inequality for at least one feasible
content. Therefore,
\[
  r_R(\theta)
  =
  \min_{m\in M(\theta)}\bar r_R(m)
  =
  \min_{m\in M(\theta)}
  \max_{\tau\in M^{-1}(m)}r_R(\tau)
  \qquad\text{for every }\theta\in\T.
\]
Thus, the skeptical reading safely implements $R$, proving condition~(i) and
the stated minimax representation.

We next prove the equivalence of conditions~(ii) and~(iii). Suppose every
proper upper contour is content-certifiable. If $\theta\in R_\ell$ with
$\ell<L$, Lemma~\ref{lem:cert} gives a content $m\in M(\theta)$ satisfying
$M^{-1}(m)\subseteq U_\ell^R$. If $\theta\in R_L$, any feasible content works
because $U_L^R=\T$. Hence, condition~(iii) holds. Conversely, suppose
condition~(iii) holds. Fix $\ell<L$ and $\theta\in U_\ell^R$. Writing
$k=r_R(\theta)\le\ell$, the content supplied by condition~(iii) satisfies
\[
  M^{-1}(m)\subseteq U_k^R\subseteq U_\ell^R.
\]
Lemma~\ref{lem:cert} then implies that $U_\ell^R$ is content-certifiable. This
proves condition~(ii).

It remains to establish the dominant-strategy statement. Suppose the
equivalent conditions hold, and let $\kappa$ be a safe reading that implements
$R$; one may take $\kappa=\bar r_R$. Define a mechanism that assigns each
submitted content $m_j$ the class $\kappa(m_j)$ and then applies the $q$-unit
chain cutoff rule. For each applicant $i$, choose
\[
  s_i(\theta)
  \in
  \arg\min_{m\in M(\theta)}\kappa(m).
\]
Implementation gives $\kappa(s_i(\theta))=r_R(\theta)$, and safety implies
that every feasible deviation is assigned a weakly larger, hence weakly worse,
class index. Since the cutoff allocation probability is weakly decreasing in
an applicant's class index, holding the other submissions fixed, $s_i$ is
weakly dominant. The induced allocation is $x^{R,q}$.

Conversely, suppose that a content-based mechanism $g$ and weakly dominant
feasible strategies $(s_1,\ldots,s_n)$ implement $x^{R,q}$. Fix an applicant
$i$ and a proper boundary $\ell<L$. Choose $\alpha\in R_1$ and
$\beta\in R_{\ell+1}$. Give the other $n-1$ applicants the profile consisting
of $q-1$ copies of type $\alpha$ and $n-q$ copies of type $\beta$, and let
$m_{-i}^\ell$ be their prescribed submissions. Define
\[
  A_\ell
  :=
  \{m\in\M:g_i(m,m_{-i}^\ell)=1\}.
\]
If the focal applicant's type belongs to $U_\ell^R$, the full-information rule
assigns her the remaining unit with probability one, so her prescribed content
belongs to $A_\ell$. If her type lies outside $U_\ell^R$, her prescribed
allocation probability is strictly below one. Were some feasible deviation to
belong to $A_\ell$, weak dominance of the prescribed strategy would require
its allocation probability to be at least one, a contradiction. Therefore,
\[
  U_\ell^R
  =
  \{\theta\in\T:M(\theta)\cap A_\ell\neq\emptyset\}.
\]
Thus, every proper upper contour is content-certifiable, proving
condition~(ii) and completing the proof.
\end{proof}

\begin{proof}[Proof of Theorem~\ref{thm:reserve-burden}]
For every content $m$,
\[
  \bar\eta_\rho(m)
  =
  \bigcap_{\tau\in M^{-1}(m)}\rho(\tau)
  =
  \{r\in\Lcal:M^{-1}(m)\subseteq D_r\}.
\]
A reserve reading $\eta$ is safe for $\rho$ if and only if
$\eta(m)\subseteq\bar\eta_\rho(m)$ for every $m$. Hence,
$\bar\eta_\rho$ is the pointwise largest safe reserve reading.

Suppose condition~(i) holds, and let $\eta$ be a safe reserve reading that
implements $\rho$. For each label $r$, define
\[
  A_r:=\{m\in\M:r\in\eta(m)\}.
\]
If $\theta\in D_r$, implementation gives some $m\in M(\theta)$ with
$r\in\eta(m)$, so $\theta$ passes $A_r$. If $\theta\notin D_r$, safety implies
that no feasible content of $\theta$ belongs to $A_r$. Thus, $A_r$ certifies
$D_r$, proving condition~(ii).

The equivalence of conditions~(ii) and~(iii) follows from
Lemma~\ref{lem:cert}. We now connect condition~(iii) to the fixed-point
representation. Since $m\in M(\theta)$ implies $\theta\in M^{-1}(m)$,
\[
  (\mathcal V\rho)(\theta)
  =
  \bigcup_{m\in M(\theta)}\bar\eta_\rho(m)
  \subseteq
  \rho(\theta)
  \qquad\text{for every }\theta.
\]
If condition~(iii) holds and $r\in\rho(\theta)$, then $\theta\in D_r$ and
there exists $m\in M(\theta)$ with $M^{-1}(m)\subseteq D_r$. Hence,
$r\in\bar\eta_\rho(m)$ and therefore $r\in(\mathcal V\rho)(\theta)$. Thus,
$\mathcal V\rho=\rho$. Conversely, if $\mathcal V\rho=\rho$, then for every
$r\in\rho(\theta)$ there is some $m\in M(\theta)$ such that
$r\in\bar\eta_\rho(m)$, equivalently $M^{-1}(m)\subseteq D_r$. Hence, the
fixed-point condition is equivalent to condition~(iii).

Under these conditions,
\[
  \rho(\theta)
  =
  \bigcup_{m\in M(\theta)}\bar\eta_\rho(m)
  \qquad\text{for every }\theta\in\T,
\]
so the skeptical reserve reading implements $\rho$ and condition~(i) follows.

For the dominant-strategy statement, fix for each type $\theta$ and each
$r\in\rho(\theta)$ a content $m_r(\theta)\in M(\theta)$ satisfying
$r\in\bar\eta_\rho(m_r(\theta))$. Let the target-attaining application set
$a_r^*(\theta)=m_r(\theta)$ when $r\in\rho(\theta)$ and
$a_r^*(\theta)=\bot$ otherwise. It certifies exactly $\rho(\theta)$. Any other
feasible application $a$ certifies a set
\[
  \{r\in\Lcal:a_r\neq\bot\text{ and }r\in\bar\eta_\rho(a_r)\}
  \subseteq
  \rho(\theta)
\]
by safety. Holding the other applicants' certified label sets fixed, label
monotonicity therefore makes $a^*(\theta)$ weakly dominant. Under these
strategies, the downstream rule receives the label profile
$(\rho(\theta_1),\ldots,\rho(\theta_n))$ and hence implements
$x^{\rho,\mathbf q}$.

Finally, suppose the institution requires one common content. If a safe
reserve reading $\eta$ implements $\rho$, then for every $\theta$ there is
$m\in M(\theta)$ such that $\eta(m)=\rho(\theta)$. Safety gives
$\eta(m)\subseteq\bar\eta_\rho(m)$, while feasibility of $m$ for $\theta$
gives $\bar\eta_\rho(m)\subseteq\rho(\theta)$. Hence,
\[
  \bar\eta_\rho(m)=\rho(\theta).
\]
Conversely, if every type has such a content, the skeptical reserve reading is
safe and that content yields exactly the type's target label set. Therefore,
$\bar\eta_\rho$ implements $\rho$ under the common-content requirement.
\end{proof}

\begin{proof}[Proof of Theorem~\ref{thm:repairs}]
We first consider chain priority systems. If $r\le r'$ pointwise, monotonicity
of the maximum and minimum operators gives $\mathcal S r\le\mathcal S r'$.
Moreover, for every $m\in M(\theta)$,
\[
  \max_{\tau\in M^{-1}(m)}r(\tau)\ge r(\theta),
\]
because $\theta\in M^{-1}(m)$. Taking the minimum over $m\in M(\theta)$ gives
$\mathcal S r\ge r$, so $\mathcal S$ is extensive.

To prove idempotence, fix $\theta$ and choose $m^*\in M(\theta)$ attaining the
minimum in $(\mathcal S r)(\theta)$. For every
$\tau\in M^{-1}(m^*)$, the same content $m^*$ is feasible for $\tau$, and
therefore
\[
  (\mathcal S r)(\tau)
  \le
  \max_{\sigma\in M^{-1}(m^*)}r(\sigma)
  =
  (\mathcal S r)(\theta).
\]
It follows that
\[
  \bigl(\mathcal S(\mathcal S r)\bigr)(\theta)
  \le
  \max_{\tau\in M^{-1}(m^*)}(\mathcal S r)(\tau)
  \le
  (\mathcal S r)(\theta).
\]
Extensivity applied to $\mathcal S r$ gives the reverse inequality. Hence,
$\mathcal S(\mathcal S r)=\mathcal S r$.

Now let $r^E=\mathcal S r^0$. Extensivity gives $r^E\ge r^0$, while
idempotence gives $\mathcal S r^E=r^E$. By
Theorem~\ref{thm:chain-burden}, the ranking induced by $r^E$ is safely
implementable by its skeptical reading. If $\widetilde r\ge r^0$ is any other
safely implementable relaxation, then $\mathcal S\widetilde r=\widetilde r$
and monotonicity gives
\[
  r^E
  =
  \mathcal S r^0
  \le
  \mathcal S\widetilde r
  =
  \widetilde r.
\]
Thus, $r^E$ is the pointwise smallest safely implementable relaxation of
$r^0$.

We next consider reserve systems. If $D\subseteq D'$, every content whose
certifier set is contained in $D$ also has its certifier set contained in
$D'$, so $\mathcal I(D)\subseteq\mathcal I(D')$. Hence, $\mathcal I$ is
monotone. Its definition gives $\mathcal I(D)\subseteq D$, so it is
contractive. Lemma~\ref{lem:cert} shows that $\mathcal I(D)$ is
content-certifiable and that a class is content-certifiable if and only if it
is a fixed point of $\mathcal I$. Therefore,
\[
  \mathcal I(\mathcal I(D))=\mathcal I(D),
\]
so $\mathcal I$ is idempotent. The same lemma gives
\[
  D_r^E=\mathcal I(D_r^0)
\]
as the largest content-certifiable subclass of $D_r^0$ for every $r$.

Define $\rho^E(\theta)=\{r\in\Lcal:\theta\in D_r^E\}$. For every type
$\theta$ and label $r$,
\[
  \begin{aligned}
  r\in(\mathcal V\rho^0)(\theta)
  &\Longleftrightarrow
  \exists m\in M(\theta)
  \text{ such that }M^{-1}(m)\subseteq D_r^0\\
  &\Longleftrightarrow
  \theta\in\mathcal I(D_r^0)=D_r^E\\
  &\Longleftrightarrow
  r\in\rho^E(\theta).
  \end{aligned}
\]
Thus, $\rho^E=\mathcal V\rho^0$. Contractiveness of $\mathcal I$ implies
$D_r^E\subseteq D_r^0$ for every $r$, and hence
$\rho^E(\theta)\subseteq\rho^0(\theta)$ for every $\theta$. Since each
$D_r^E$ is content-certifiable, Theorem~\ref{thm:reserve-burden} implies that
$\rho^E$ is safely implementable by the skeptical reading
$\bar\eta_{\rho^E}$; equivalently, $\mathcal V\rho^E=\rho^E$.

Finally, let $\widetilde\rho$ be any evidence-safe contraction of $\rho^0$,
and define
\[
  \widetilde D_r
  :=
  \{\theta\in\T:r\in\widetilde\rho(\theta)\}.
\]
By Theorem~\ref{thm:reserve-burden}, each $\widetilde D_r$ is
content-certifiable, while contraction gives
$\widetilde D_r\subseteq D_r^0$. Maximality of $D_r^E$ therefore yields
$\widetilde D_r\subseteq D_r^E$ for every $r$, or equivalently,
\[
  \widetilde\rho(\theta)\subseteq\rho^E(\theta)
  \qquad\text{for every }\theta\in\T.
\]
Thus, $\rho^E$ is the pointwise largest evidence-safe contraction of
$\rho^0$. If some values in the range of $r^E$ are unused, deleting the
corresponding empty classes and relabeling the remaining classes preserves the
induced ranking and allocation.
\end{proof}

\begin{proof}[Proof of Corollary~\ref{cor:signed}]
Under signed-bundle evidence,
\[
  M^{-1}(P,A)=\T(P,A).
\]
Part~(i) therefore follows directly from Theorem~\ref{thm:chain-burden}:
type $\theta$ can attain its target class safely if and only if it has some
feasible bundle satisfying
\[
  \T(P,A)\subseteq U_{r_R(\theta)}^R.
\]
Because $\theta\in\T(P,A)$ whenever $(P,A)\in G(\theta)$, the inner maximum in
the skeptical expression is at least $r_R(\theta)$; equality is attainable
exactly under the stated containment condition.

For part~(ii), condition~(iii) of Theorem~\ref{thm:reserve-burden} becomes: for
every label $r$ and every $\theta\in D_r$, there exists
$(P,A)\in G(\theta)$ such that
\[
  \T(P,A)\subseteq D_r.
\]
This is exactly the stated condition, and Lemma~\ref{lem:cert} gives the
equivalent union representation.
\end{proof}

\begin{proof}[Proof of Proposition~\ref{prop:onesided}]
We first consider chain priority systems. Under positive evidence, every
certificate feasible for $\theta$ is also feasible for every
$\theta'\supseteq\theta$. Thus,
\[
  \theta\subseteq\theta'
  \quad\Longrightarrow\quad
  \theta\preceqM\theta'.
\]
Every safely implementable upper contour is content-certifiable, so
Proposition~\ref{prop:fullreport} implies $\theta'\succeq_R\theta$. Hence,
inclusion monotonicity is necessary. Conversely, under rich positive evidence, type $\theta$ can submit
$(\theta,\emptyset)$, whose certifier set is
\[
  \T(\theta,\emptyset)
  =
  \{\tau\in\T:\theta\subseteq\tau\}.
\]
If $R$ is inclusion-monotone, this set is contained in
$U_{r_R(\theta)}^R$, so Corollary~\ref{cor:signed}(i) gives safe
implementation.

Under negative evidence,
\[
  \theta\subseteq\theta'
  \quad\Longrightarrow\quad
  M(\theta')\subseteq M(\theta),
\]
so $\theta'\preceqM\theta$. Safe implementation therefore requires
$\theta\succeq_R\theta'$, which is inclusion antitonicity. Conversely, under
rich negative evidence, type $\theta$ can submit
$(\emptyset,\C\setminus\theta)$, whose certifier set is
\[
  \T(\emptyset,\C\setminus\theta)
  =
  \{\tau\in\T:\tau\subseteq\theta\}.
\]
If $R$ is inclusion-antitone, this set is contained in
$U_{r_R(\theta)}^R$, so Corollary~\ref{cor:signed}(i) again gives safe
implementation. Under exact-profile evidence, type $\theta$ can submit
$(\theta,\C\setminus\theta)$, whose certifier set is $\{\theta\}$, so the
same condition holds for every target ranking.

We next consider reserve systems. Under positive evidence, every content
feasible for $\theta$ is feasible for every $\theta'\supseteq\theta$.
Because every implementable eligibility class is content-certifiable,
Proposition~\ref{prop:fullreport} implies
\[
  \theta\in D_r,\ \theta\subseteq\theta'
  \quad\Longrightarrow\quad
  \theta'\in D_r.
\]
Thus, inclusion monotonicity of $\rho$ is necessary. Conversely, if $\rho$ is
inclusion-monotone, then for every $\theta\in D_r$, the rich positive bundle
$(\theta,\emptyset)$ has certifier set
\[
  \{\tau\in\T:\theta\subseteq\tau\}\subseteq D_r.
\]
Corollary~\ref{cor:signed}(ii) therefore gives implementation.

Under negative evidence, the argument is reversed. If
$\theta\subseteq\theta'$, every negative content feasible for $\theta'$ is
feasible for $\theta$, so implementation requires
$\rho(\theta')\subseteq\rho(\theta)$. Conversely, if $\rho$ is
inclusion-antitone, then for every $\theta\in D_r$, the bundle
$(\emptyset,\C\setminus\theta)$ has certifier set
\[
  \{\tau\in\T:\tau\subseteq\theta\}\subseteq D_r.
\]
Corollary~\ref{cor:signed}(ii) gives implementation. Under exact-profile
evidence, every $\theta\in D_r$ has a certificate with certifier set
$\{\theta\}\subseteq D_r$, so every target eligibility map is implementable.
\end{proof}

\begin{proof}[Proof of Proposition~\ref{prop:nestedness}]
If $\mathcal D(\rho)$ is empty, the result is immediate. Necessity follows
because any two upper contours of one chain are nested. For sufficiency, order
the distinct nonempty classes as
$D^1\subsetneq\cdots\subsetneq D^K$ and define
$R_1=D^1$ and $R_k=D^k\setminus D^{k-1}$ for $k=2,\ldots,K$. If
$D^K\neq\T$, append $R_{K+1}=\T\setminus D^K$. These sets form an ordered
partition of $\T$, and its $k$th upper contour is
$\bigcup_{j=1}^kR_j=D^k$ for every $k\le K$.
\end{proof}

\begin{proof}[Proof of Proposition~\ref{prop:divergence}]
If $D_r$ is not an upper contour, there exist $\theta\in D_r$ and
$\tau\notin D_r$ with $\tau\succeq_R\theta$; otherwise $D_r$ is upward closed
and equals $U_\ell^R$ for
$\ell=\max_{\theta\in D_r}r_R(\theta)$. Consider a two-applicant profile with
these types and a one-unit system with $r$ as its only active reserve and no
open capacity. Reserve feasibility implies that $\tau$ receives probability
zero. Under the chain, $\tau$ receives probability one if
$\tau\succ_R\theta$ and probability one-half if $\tau\sim_R\theta$.
The allocations therefore differ.
\end{proof}

\section{Examples of rules -- Chain priorities}
\label{app:chain-examples}

Fix an ordered list of $K$ priority criteria. For each criterion $k$, let
\[
  z_k(\theta)\in\{0,1\}
\]
indicate whether type $\theta$ satisfies that criterion in the
priority-favored direction. For a favorable category $C_k$, for example, one
may take
\[
  z_k(\theta)=\mathbf{1}\{C_k\in\theta\},
\]
whereas a clearance criterion may be represented by
\[
  z_k(\theta)=\mathbf{1}\{C_k\notin\theta\}.
\]

\paragraph{Highest-priority-criterion rule.}
Suppose the criteria are ordered from highest to lowest priority. Define
\[
  h(\theta)
  =
  \min\{k:z_k(\theta)=1\},
\]
with $h(\theta)=K+1$ if type $\theta$ satisfies no criterion. Type $\theta'$
has at least as much priority as type $\theta$ whenever
\[
  h(\theta')\le h(\theta).
\]
Thus, an applicant is ranked according to the highest-priority criterion she
satisfies; lower-ranked criteria do not affect her position.

\paragraph{Lexicographic rule.}
Associate each type with the vector
\[
  z(\theta)
  =
  \bigl(z_1(\theta),\ldots,z_K(\theta)\bigr).
\]
Two types are compared at the first coordinate at which their vectors differ.
At that coordinate, the type with value one is ranked above the type with
value zero. Lower-ranked criteria matter only when all higher-ranked criteria
are tied.

\paragraph{Additive scoring rule.}
Let $w_1,\ldots,w_K>0$ be the weights assigned to the criteria, and define
\[
  s_w(\theta)
  =
  \sum_{k=1}^K w_k z_k(\theta).
\]
A type with a higher score receives higher priority, while types with the same
score are tied.

Each of these rules induces an ordered ranking of the finite type space.
Grouping types that are tied under the relevant rule produces an ordered
partition
\[
  R=(R_1,\ldots,R_L),
\]
which is the chain priority system studied in the main text. The three rules
differ in how they aggregate applicant characteristics: the
highest-priority-criterion rule retains only the best satisfied criterion, the
lexicographic rule applies hierarchical tie-breaking, and the scoring rule
permits numerical trade-offs across criteria.

\end{appendices}


\begin{thebibliography}{99}

\bibitem{BenPorathDekelLipman2014}
Ben-Porath, E., E. Dekel, and B. L. Lipman (2014).
``Optimal Allocation with Costly Verification.''
\emph{American Economic Review} 104(12), 3779--3813.

\bibitem{BenPorathDekelLipman2019}
Ben-Porath, E., E. Dekel, and B. L. Lipman (2019).
``Mechanisms with Evidence: Commitment and Robustness.''
\emph{Econometrica} 87(2), 529--566.

\bibitem{BenPorathLipman2012}
Ben-Porath, E., and B. L. Lipman (2012).
``Implementation with Partial Provability.''
\emph{Journal of Economic Theory} 147(5), 1689--1724.

\bibitem{BullWatson2007}
Bull, J., and J. Watson (2007).
``Hard Evidence and Mechanism Design.''
\emph{Games and Economic Behavior} 58(1), 75--93.

\bibitem{DurKominersPathakSonmez2018}
Dur, U., S. D. Kominers, P. A. Pathak, and T. S\"onmez (2018).
``Reserve Design: Unintended Consequences and the Demise of Boston's Walk
Zones.''
\emph{Journal of Political Economy} 126(6), 2457--2479.

\bibitem{EchenicheYenmez2015}
Echenique, F., and M. B. Yenmez (2015).
``How to Control Controlled School Choice.''
\emph{American Economic Review} 105(8), 2679--2694.

\bibitem{GreenLaffont1986}
Green, J. R., and J.-J. Laffont (1986).
``Partially Verifiable Information and Mechanism Design.''
\emph{Review of Economic Studies} 53(3), 447--456.

\bibitem{HafalirYenmezYildirim2013}
Hafalir, I. E., M. B. Yenmez, and M. A. Yildirim (2013).
``Effective Affirmative Action in School Choice.''
\emph{Theoretical Economics} 8(2), 325--363.

\bibitem{HagenbachKoesslerPerezRichet2014}
Hagenbach, J., F. Koessler, and E. Perez-Richet (2014).
``Certifiable Pre-Play Communication: Full Disclosure.''
\emph{Econometrica} 82(3), 1093--1131.

\bibitem{KoesslerPerezRichet2019}
Koessler, F., and E. Perez-Richet (2019).
``Evidence Reading Mechanisms.''
\emph{Social Choice and Welfare} 53(3), 375--397.

\bibitem{KominersSonmez2016}
Kominers, S. D., and T. S\"onmez (2016).
``Matching with Slot-Specific Priorities: Theory.''
\emph{Theoretical Economics} 11(2), 683--710.

\bibitem{LipmanSeppi1995}
Lipman, B. L., and D. J. Seppi (1995).
``Robust Inference in Communication Games with Partial Provability.''
\emph{Journal of Economic Theory} 66(2), 370--405.

\bibitem{PathakSonmezUnverYenmez2024}
Pathak, P. A., T. S\"onmez, M. U. \"Unver, and M. B. Yenmez (2024).
``Fair Allocation of Vaccines, Ventilators and Antiviral Treatments: Leaving
No Ethical Value Behind in Health Care Rationing.''
\emph{Management Science} 70(6), 3999--4036.

\bibitem{PerezRichetSkreta2026}
Perez-Richet, E., and V. Skreta (2026).
``Falsification-Proof Non-Market Allocation Mechanisms.''
Working paper, April 2026.

\bibitem{SonmezYenmez2022}
S\"onmez, T., and M. B. Yenmez (2022).
``Affirmative Action in India via Vertical, Horizontal, and Overlapping
Reservations.''
\emph{Econometrica} 90(3), 1143--1176.

\end{thebibliography}
\end{document}